\documentclass[aps,pra,10pt,a4paper,reprint,nofootinbib,superscriptaddress]{revtex4-2}
\usepackage[utf8]{inputenc}
\usepackage{amsfonts}
\usepackage{amssymb}
\usepackage{amsmath}
\usepackage{amsthm}
\usepackage{graphics}
\usepackage{graphicx}
\usepackage{braket}
\usepackage[colorlinks=true,linkcolor=blue,urlcolor=blue,citecolor=blue]{hyperref}
\usepackage{times,txfonts}

\newtheorem{definition}{Definition}
\newtheorem{proposition}{Proposition}

\newcommand{\ketbra}[2]{|#1\rangle \langle #2|}

\begin{document}
\title{Identifying the value of a random variable unambiguously: Quantum versus classical approaches}

\author{Saronath Halder}
\affiliation{Centre for Quantum Optical Technologies, Centre of New Technologies, University of Warsaw, Banacha 2c, 02-097 Warsaw, Poland}

\author{Alexander Streltsov}
\affiliation{Centre for Quantum Optical Technologies, Centre of New Technologies, University of Warsaw, Banacha 2c, 02-097 Warsaw, Poland}
\affiliation{Institute of Fundamental Technological Research, Polish Academy of Sciences, Pawińskiego 5B, 02-106 Warsaw, Poland}

\author{Manik Banik}
\affiliation{Department of Physics of Complex Systems, S.N. Bose National Center for Basic Sciences, Block JD, Sector III, Salt Lake, Kolkata 700106, India}

\begin{abstract}
Quantum resources may provide advantage over their classical counterparts. Theoretically, in certain tasks, this advantage can be very high. In this work, we construct such a task based on a game, mediated by Referee and played between Alice and Bob. Referee sends Alice a value of a random variable. At the same time, Referee also sends Bob some partial information regarding that value. Here partial information can be defined in the following way. Bob gets the information of a random set which must contain the value of the variable, that is sent to Alice by the Referee, along with other value(s). Alice is not allowed to know what information is sent to Bob by the Referee. Again, Bob does not know which value of the random variable is sent to Alice. Now, the game can be won if and only if Bob can unambiguously identify the value of the variable, that is sent to Alice, with some nonzero probability, no matter what information Bob receives or which value is sent to Alice. However, to help Bob, Alice sends some limited amount of information to him, based on any strategy which is fixed by Alice and Bob before the game begins. We show that if Alice sends limited amount of classical information then the game cannot be won while the quantum analogue of the `limited amount of classical information' is sufficient for winning the game. Thus, it establishes quantum advantage. We further analyze several variants of the game and provide certain bounds on the success probabilities. Moreover, we establish connections between trine ensemble, mutually unbiased bases, and the encoding-decoding strategies of those variants. We also discuss the role of quantum coherence in the present context. 
\end{abstract}
\maketitle

\section{Introduction}
Efficient utilization of non-classical features of elementary quantum systems, such as coherent superposition, quantum entanglement, measurement incompatibility, and indefinite causal order leads to advantageous information and communication protocols that otherwise are not possible with classical resources \cite{Dowling03, Deutsch20}. A few innovative such protocols are quantum cryptography \cite{Ekert91}, quantum superdense coding \cite{Bennett92}, and quantum teleportation \cite{Bennett93} that establish quantum advantages in communication scenario by invoking quantum entanglement between the sender and the receiver. Quantum advantages, however, are hard to find and sometimes constrained by fundamental no-go theorems. For instance, Holevo's no-go theorem \cite{Holevo73} limits the capacity of a quantum channel as that of its classical counterpart when no preshared entanglement between the sender and the receiver is allowed. More recently, a stronger version of this no-go theorem has been obtained which establishes that the classical information storage in an $n$-level quantum system is not better than the corresponding classical $n$-state system \cite{Frenkel15}.

In this work, we report a novel communication advantage of an elementary quantum system without invoking any preshared entanglement between the sender (Alice) and the receiver (Bob). At this point the task of random access codes (RAC) is worth mentioning which also depicts communication advantages of quantum systems between unentangled sender and receiver. In RAC a long message is encoded into fewer bits with the ability to recover (decode) any one of the initial bits with high degree of success probability. Historically quantum random access codes (QRAC) were first studied by Wiesner by the name `conjugate coding' \cite{Wiesner83}. Later it was re-analyzed by Ambainis {\it et al.} \cite{Ambainis99, Ambainis02} and subsequently draws a huge research interest \cite{Aaronson05, Gavinsky05, Aaronson07, Spekkens09, Pawowski09, Li12, Banik15, Tavakoli15, Ambainis19, Vaisakh21, Patra22}. The task we consider, however, is different than the RAC task and can be best described in terms of a game, mediated by Referee and played between Alice and Bob. Referee sends Alice a value of a random variable. At the same time, Referee also sends Bob some partial information regarding that value. Here partial information can be defined in the following way. Bob gets the information of a random set which must contain the value of the variable, that is sent to Alice by the Referee, along with other value(s). Alice is not allowed to know what information is sent to Bob by the Referee. Again, Bob does not know which value of the random variable is sent to Alice. Now, the game can be won if and only if Bob can unambiguously identify the value of the variable, that is sent to Alice, with some nonzero probability, no matter what information Bob receives or which value is sent to Alice. However, to help Bob, Alice sends some limited information to him. This is based on any strategy which is fixed by Alice and Bob before the game begins. We mention here that only deterministic strategies are considered in this work, i.e., the parties do not use any randomness. For example, when Alice sends a cbit, she sends either `0' or `1'. On the other hand, when she sends a qubit, she actually sends a pure qubit state. No additional correlation (local or global) is used by Alice and Bob. See also Fig.~\ref{fig1} for the description of the present game. 

\begin{figure*}
\includegraphics[scale=0.4]{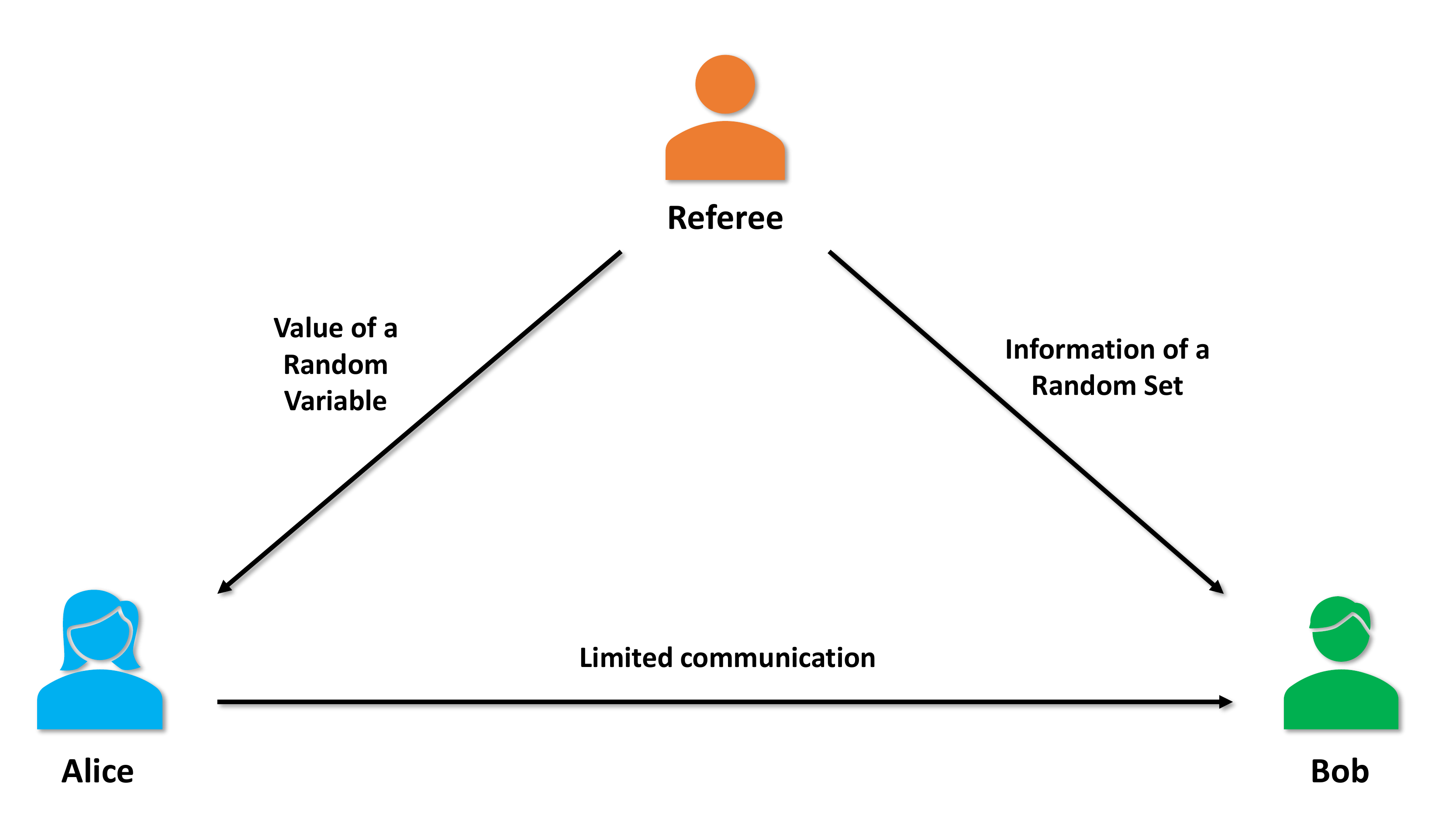}
\caption{There are three spatially separated parties: Referee, Alice, and Bob. Referee sends a value of a random variable to Alice. At the same time, Referee sends the information of a random set to Bob. The set contains the value of the random variable, sent to Alice, along with some other value(s). Remember that Alice does not know about the information of the random set which is sent to Bob but she knows the size of the set. Similarly, Bob does not know which value of the random variable is sent to Alice. The task of Bob is to unambiguously identify the value of the random variable, sent to Alice, with some nonzero probability all the time, i.e., no matter what information he receives from the Referee or which value of the random variable is sent to Alice by the Referee. However, Alice is allowed to send a limited amount of information to Bob based on any pre-decided strategy. This is to help Bob in identifying the value of the variable unambiguously. Note that both Alice and Bob know about all possible values of the random variable.}\label{fig1}
\end{figure*}

Note that the game can be won perfectly if there is no restriction imposed on the available communication from Alice to Bob. Interesting situations arise only when the allowed communication is limited. In that situation, Bob is not able to identify the value of the random variable perfectly all the time. Then, it can be explored, how well Bob can identify the value of the random variable no matter what information he receives. In particular, this becomes a probabilistic case which helps us to explore the advantages and limitations of resources. In this direction (the direction of the probabilistic study) there are two popular settings, researchers usually adopt. One is the minimum error strategy: Bob can try to identify the value of the random variable minimizing the error. The other is the unambiguous strategy: Bob can try to identify the value of the random variable without committing any error but in this case there will be a nonzero probability of inconclusive outcome. In other words, the unambiguous strategy can be explained as: either answering the right result without any error, or answering `inconclusive'. But the former answer has to occur with non-zero probability. In this work, we consider the second strategy, i.e., the unambiguous strategy and explore corresponding bounds on the success probabilities, considering different cases. We also explore the role of several mathematical concepts in these cases. 

We are now ready to define the task of Bob more accurately. This is done by following the definition of unambiguously distinguishable set of quantum states. Suppose, a set of quantum states is given and we want to distinguish these states unambiguously. If a particular state of the given set can be identified error-freely with some nonzero probability then we say that the state is unambiguously identifiable. Moreover, if all the states of a given set are unambiguously identifiable then the set is unambiguously distinguishable \cite{Chefles04, Halder22} and the task of the state distinguishability can be accomplished unambiguously with some nonzero probability. Similarly, here we are interested in those situations where all values of a random variable are unambiguously identifiable, no matter what information Bob receives. Such a situation implies that the task of determining the value of the random variable unambiguously with some nonzero probability can be accomplished. So, we set the condition of winning the game as: the game can be won if and only if Bob is able to identify the value of the random variable error-freely with some nonzero probability all the time (no matter which value is sent to Alice by the Referee or what information Bob receives). 

Previously, a few communication tasks have been designed where huge separation between classical and quantum resources has been reported \cite{Massar01, Buhrman01, Wolf03, Galvao03, Perry15}. In our game also, we report huge advantage of qubit communication over classical communication and theoretically, this advantage might be increased up to an arbitrary height if the dimension of the random variable increases. In Sec.~\ref{sec2}, we first present an elementary version of the game. We also present several variants of this game in this section. We find the connection of trine ensemble with the encoding-decoding strategy of a variant. Then, we provide several generalizations of this game for higher dimensional random variable. These generalizations are given in Sec.~\ref{sec3} and in Sec.~\ref{sec4}. Eventually, we present several bounds on the success probabilities and find a connection with mutually unbiased bases and the encoding-decoding strategy of a variant of the game. In a few cases, we derive the optimal success probabilities and discuss about their achievability. In Sec.~\ref{sec5}, another generalization is given. We also discuss the role of coherence in this game. Finally, the conclusion is drawn in Sec.~\ref{sec6} mentioning some open problems for further research.

\section{An elementary version of the game}\label{sec2}
{\it Description.}~We assume that there are three parties, Referee, Alice, and Bob. Referee sends the value of a three-dimensional random variable to Alice. We denote the variable by $X$ and its dimension by $d$. In this section, we assume $d =3$. Here $X$ is a discrete variable, so, one may formulate our communication task without defining the dimension of $X$. However, for the convenience of equations that will appear later in the paper, we define the dimension of the variable X as the number of available values. In fact, in some cases, this quantity also helps us to understand the quality or quantity of the communication required (from Alice to Bob) to accomplish our task. 

So, Referee sends Alice $x_i$ (value of $X$), while $x_i$ belongs to the set $\{x_1, x_2, x_3\}$. This set is known to both Alice and Bob. On the other hand, Referee sends Bob `$j$', where this `$j$' is associated with a random set $S_j$. We denote by $n:=|S_j|$ the size of $S_j$. Here in this section we consider $n = 2$. This value is known to both Alice and Bob. Depending on the values of $d$ and $n$, three random sets can be defined: $S_1\equiv\{x_1, x_2\},~ S_2\equiv\{x_2, x_3\},~ S_3\equiv\{x_3, x_1\}$. These definitions are also known to both Alice and Bob. Note that two things, (i) sending $x_i$ to Alice and (ii) sending `$j$' to Bob, are simultaneously done by Referee. Again, information of which set is sent to Bob, is not known to Alice and similarly, which value of the random variable is sent to Alice, is not known to Bob. But $j$ will be chosen (randomly) in such a way that $S_j$ must contain the particular value of $X$ which is sent to Alice by the Referee, along with some other value of $X$. For example, suppose $x_1$ is sent to Alice then either 1 or 3 is sent to Bob. If 1 is sent then it means Bob is instructed that the value of the variable, which is sent to Alice, belongs to $S_1$. Similarly, if 3 is sent then it means Bob is instructed that the value of the variable, which is sent to Alice, belongs to $S_3$. Clearly, before receiving the value of the random variable by Alice, Alice and Bob know that the information regarding any set can be sent to Bob and these sets are equally probable. When we say that the sets are equally probable, one may point out that this applies only to the sets containing the value of $X$ which is sent to Alice by the Referee because the probability of the other set is null. While this is correct, in the overall process, all sets are equally probable. Therefore, Alice and Bob have to fix an encoding-decoding strategy accordingly. To help Bob in identifying the value of the variable, Alice is allowed to send a classical bit (cbit) or a quantum bit (qubit). The game can be won if and only if Bob is able to identify $x_i$ unambiguously with some nonzero probability $\forall j =1, 2, 3$. Remember that Alice does not know the information of which set is sent to Bob, but she knows that the set must contain the value of the random variable, which she has received, along with some other value. 

In this context, the first thing we want to prove is the following. If Alice sends only a cbit to Bob then it is not always possible for Bob to unambiguously identify the value of the random variable, which Alice receives, with some nonzero probability. The intuition of this can be found through a simple example. Suppose, Alice tries to fix an encoding strategy and for this purpose she thinks about computing a function:
\begin{equation*}
\begin{array}{c}
\mathcal{F} = 0 ~~~ \mbox{if} ~~~ x_i = x_1,\\[1 ex]
\mathcal{F} = 1 ~~~ \mbox{if} ~~~ x_i \neq x_1.
\end{array}  
\end{equation*}
This is one of the simplest forms that can be computed by Alice and before the game starts, Alice can inform about computing this function to Bob. Similarly, there can be many other strategies which can be adopted by the duo, Alice and Bob. However, the key point is that $\mathcal{F}$ cannot have more than two values since Alice is allowed to send a cbit only to Bob. But through those two values of $\mathcal{F}$ it is not possible for Bob to extract three different values of $X$, which is necessary for unambiguous identification. Precisely, in this example, if $x_2$ or $x_3$ is sent to Alice, then She sends `1' to Bob. In such a situation if Bob receives $j = 2$ value from the Referee, he is not be able to identify the value of the random variable unambiguously with some nonzero probability.

We are now ready to present the above observation in a proposition form and we also provide a general proof of the proposition.

\begin{proposition}\label{prop1}
There exists no strategy via which the game can be won when Alice is allowed to send only a cbit to Bob.
\end{proposition}

\begin{proof}
The encoding-decoding strategy should be fixed before the value of the random variable is sent to Alice. The values of the variable are equally probable. Thus, the information of which set is sent to Bob, is also completely random. Furthermore, this information is not known to Alice. Thus, `the information regarding the random set' does not help Alice to fix an encoding strategy. But this information may help Bob to choose the right decoding strategy. 

However, to fix an encoding-decoding strategy via which the game can be won, Alice has to compute a function and before the game starts Alice can inform about this function to Bob. This is like the example given in the above. But the function can be any function.

Since, the random variable is three dimensional, the function must output three different values corresponding to the values of the random variable. But when Alice is allowed to send only a cbit, Alice cannot encode three different values of a function within a cbit.

Next, we assume that the function does not output different values corresponding to the values of the random variable. In that case, there must be at least one situation when the value of the variable cannot be identified unambiguously. More precisely, it means that there exists at least one value of $j$ (where $j$ is associated with $S_j$), for which the unambiguous identification is not possible. These complete the proof.  
\end{proof}

\begin{proposition}\label{prop2}
There exists a strategy via which the game can be won when Alice sends a qubit to Bob, i.e., Bob is able to identify $x_i$ $\forall~j$ unambiguously with some nonzero probability.
\end{proposition}

\begin{proof}
To prove the above proposition, if an explicit strategy is provided via which the game can be won then it is sufficient. For example, Alice can avail herself of the following encoding strategy:
\begin{equation*}
\begin{array}{c}
x_1\rightarrow\ket{0},~~x_2\rightarrow\ket{1},\\[1 ex]
x_3\rightarrow\frac{1}{\sqrt{2}}\ket{0}+\frac{1}{\sqrt{2}}\ket{1},
\end{array}
\end{equation*} 
where the states $\ket{0}$ and $\ket{1}$ are orthogonal to each other.

Now, if any two values of $X$ are chosen, then corresponding states are linearly independent states which can be unambiguously distinguished with some nonzero probability \cite{Chefles98}. In this way, for any value of $j$ (where $j$ is associated with $S_j$), Bob is able to identify the value of the variable unambiguously with at least some nonzero probability and thus, the game can be won. These complete the proof. 
\end{proof}

From Proposition \ref{prop1} and \ref{prop2}, it is clear that in the context of the present game a qubit can provide advantage over a cbit. Again, the advantage is coming from superposition. Also notice that we compare between deterministic strategies only, since, we do not allow the parties to use any type of randomness here. More precisely, in our case when Alice sends a cbit, she sends either `0' or `1' and when Alice sends a qubit, she actually sends a pure qubit state. No additional correlation (local or global) is used by Alice and Bob. However, in the following, we explore about the advantage of qubit communication in details. 

\subsection{Bounds on the success probabilities}
Before we proceed, we mention that we are going to talk about two types of success probabilities: {\it individual probability} and {\it average probability}. But before we provide these definitions it is important to say the following. Here we consider a random variable which has different values. Now a value of the variable is sent to Alice, at the same time Referee also sends some partial information regarding that value to Bob. This partial information is defined by different $j$ values, where $j$ is associated with $S_j$. Corresponding to each $j$ value, we define `{\it an event}'. Based on these events, we now provide the definitions of two types of success probabilities.

\begin{definition}\label{}
[Individual probability of success] It is the probability of successfully identifying the value of the random variable, which is sent to Alice, in each event.
\end{definition}

According to the winning condition of the game, the individual probabilities of success must be non-zero.

\begin{definition}\label{}
[Average probability of success] We first take the sum of all individual probabilities. Then, we divide that sum by the total number of events. Thus, we get the average probability of success. 
\end{definition}

As mentioned, an event is defined by a `$j$' value. If $p_j$ is the probability of success corresponding to a `$j$' value. Then, $p_j$s are individual probabilities of success. Furthermore, average probability of success $\mathcal{P}_{avg}^{(d)}$, when $\dim X = d$, is defined as:

\begin{center}
$\mathcal{P}_{avg}^{(d)}$ = $\frac{1}{|j|}\sum_j p_j$,
\end{center}
where $|j|$ is the total number of events, given by $\binom{d}{n}$.

\begin{proposition}\label{prop3}
The average probability of success can be maximized by sending only a cbit from Alice to Bob but in this scenario, the goal of the present game cannot be achieved.
\end{proposition}

\begin{proof}
The values of the random variable are equally probable. Thus, the information regarding any set can be sent to Bob. These sets are also equally probable. The values of the random variable of a particular set must be associated with linearly independent states. This is to ensure the unambiguous discrimination \cite{Chefles98} for achieving the present goal. In an unambiguous discrimination of two pure states, the probability of inconclusive outcome depends on the overlap of the states \cite{Ivanovic87, Dieks88, Peres88}. So, the average probability of success, denoted by $\mathcal{P}_{avg}^{(3)}$, is given as the following: 
\begin{equation}\label{eq1}
\begin{array}{c}
\mathcal{P}_{avg}^{(3)} = 1 - \frac{1}{3} [|\langle\phi_1|\phi_2\rangle| + |\langle\phi_2|\phi_3\rangle| + |\langle\phi_3|\phi_1\rangle|].
\end{array}
\end{equation}
Here the superscript `(3)' stands for the fact that the dimension of $X$ is 3. The states $\ket{\phi_i}$, $\forall i =1,2,3$, are the states which are used for the encoding strategy by Alice. They are defined as the following. We assume that the values of the random variable is mapped against these states: $x_1\rightarrow\ket{\phi_1}$, $x_2\rightarrow\ket{\phi_2}$, $x_3\rightarrow\ket{\phi_3}$, where $\ket{\phi_3} = a_1\ket{\phi_1^\perp}+a_2\ket{\phi_2^\perp}$. The states $\ket{\phi_1}$ and $\ket{\phi_2}$ must be linearly independent. The coefficients $a_1$ and $a_2$ are some complex numbers such that $\ket{\phi_3}$ is a valid quantum state. Again, we take the values of $|a_1|, |a_2|$ as nonzero. It is quite clear now that the average probability of inconclusive outcome in the present case is dependent on $[|\langle\phi_1|\phi_2\rangle|+|\langle\phi_2|\phi_3\rangle|+|\langle\phi_3|\phi_1\rangle|]$. So, to increase $\mathcal{P}_{avg}^{(3)}$, we have to decrease the average probability of inconclusive outcome. For this purpose, we consider the following:
\begin{equation}\label{eq2}
\begin{array}{l}
\ket{\phi_1} = \ket{\phi},~~
\ket{\phi_2} = a\ket{\phi}+b\ket{\phi^\perp},\\[1 ex]
\ket{\phi_3} = a_1\ket{\phi^\perp}+a_2(b^\ast\ket{\phi}-a^\ast\ket{\phi^\perp}).
\end{array}
\end{equation}
$\ket{\phi_1}$ and $\ket{\phi_2}$ are linearly independent, $\langle\phi|\phi^\perp\rangle = 0$ and $|a|^2+|b|^2=1$. $a^\ast$ and $b^\ast$ are complex conjugate of the complex numbers $a$ and $b$. We can rewrite $\ket{\phi_3}$ as $a_2b^\ast\ket{\phi} + (a_1-a_2a^\ast)\ket{\phi^\perp}$, where $|a_2b^\ast|^2 + |(a_1-a_2a^\ast)|^2 = 1$. We next want to calculate the lower bound of the quantity, $[|\langle\phi_1|\phi_2\rangle|+|\langle\phi_2|\phi_3\rangle|+|\langle\phi_3|\phi_1\rangle|]$, which can be rewritten as the following: 
\begin{equation}\label{eq3}
\begin{array}{c}
|\langle\phi_1|\phi_2\rangle|+|\langle\phi_2|\phi_3\rangle|+|\langle\phi_3|\phi_1\rangle| 
= |a| + |b|(|a_1| + |a_2|).
\end{array}
\end{equation}
Since, $|a_1|+|a_2|$ cannot be zero, putting $|b|=0$, we minimize the quantity $|b|(|a_1|+|a_2|)$. This implies that $[|\langle\phi_1|\phi_2\rangle|+|\langle\phi_2|\phi_3\rangle|+|\langle\phi_3|\phi_1\rangle|] = 1$ and $\mathcal{P}^{(3)}_{avg} = \frac{2}{3}$. In the following, we prove that this is the maximum value of $\mathcal{P}^{(3)}_{avg}$. But if $|b|=0$, then $\ket{\phi_1}$ and $\ket{\phi_2}$ become linearly dependent and this is not good when one wants to achieve the goal of the present game. Therefore, we want to check when $|b|\neq0$, how to reduce the value of the quantity of (\ref{eq3}), i.e., we want to check if there is any way of maximizing the quantity of (\ref{eq1}) along with winning the game.

If $|b|\neq0$ then to reduce the value of the quantity of (\ref{eq3}), we have to reduce the quantity $|a_1|+|a_2|$. It is easy to show that $(|a_1|+|a_2|)\geq1$. For this purpose, we derive the following: 
\begin{equation}\label{eq4}
\begin{array}{l}
~~~~~|a_2b^\ast|^2 + |(a_1-a_2a^\ast)|^2 = 1\\[1 ex]

\Rightarrow |a_2|^2|b|^2 + (a_1^\ast - a_2^\ast a)(a_1-a_2a^\ast) = 1\\[1 ex]

\Rightarrow |a_1|^2 + |a_2|^2 - (a_1^\ast a_2a^\ast + a_1a_2^\ast a) = 1\\[1 ex]

\Rightarrow |a_1| + |a_2| = \sqrt{[1+2|a_1||a_2|+(a_1^\ast a_2a^\ast + a_1a_2^\ast a)]} \\[1 ex]

= \sqrt{[1+2|a_1||a_2| \{1+|a|\cos{(\theta_1+\theta-\theta_2)}\}]},
\end{array}
\end{equation}
where $a = |a|e^{\mathbf{i}\theta}$ and $a_i = |a_i|e^{\mathbf{i}\theta_i}$, $\forall i = 1,2$, $\mathbf{i} = \sqrt{-1}$. From the above, it is clear that $(|a_1| + |a_2|) \geq 1$ and $(|a_1| + |a_2|) = 1$ if and only if one of the following conditions is satisfied: $|a_1| = 0$, $|a_2| = 0$, or, $|a| = 1$ along with $\cos{(\theta_1+\theta-\theta_2)} = -1$. By putting $(|a_1| + |a_2|) \geq 1$, we get the following lower bound:
\begin{equation}\label{eq5}
[|\langle\phi_1|\phi_2\rangle|+|\langle\phi_2|\phi_3\rangle|+|\langle\phi_3|\phi_1\rangle|] \geq (|a|+|b|). 
\end{equation}
We further can think about minimizing the quantity $|a|+|b|$ which is 1 if and only if either of the conditions is satisfied: $|a| = 0$ or $|b| = 0$. Finally, we consider all the possibilities together for increasing the value of $\mathcal{P}_{avg}^{(3)}$. But we see that the maximum value of this quantity is $\frac{2}{3}$. This is achievable if and only if one of following conditions is satisfied:
\begin{equation}\label{eq6}
\begin{array}{ccc}
(1) & |a| = 1,~|b| = 0, & |(a_1-a_2a^\ast)|^2 = 1\\[1 ex]
(2) & |b| = 1,~|a| = 0, & ~|a_1| = 1,~|a_2| = 0\\[1 ex]
(3) & |b| = 1,~|a| = 0, & ~|a_1| = 0,~|a_2| = 1
\end{array}
\end{equation}
It is easy to check that for each of the above conditions, two of the states of (\ref{eq2}) are going to be the same state and the other state is orthogonal to that state. Such an encoding can be communicated, for sure, from Alice's side to Bob's side by sending a cbit only, the same states can correspond to 0 while the orthogonal state can correspond to 1 or, vice-versa. However, for such an encoding, for at least one value of $j$ ($j$ is associated with $S_j$), Bob will not be able to identify the value of $X$ unambiguously with some nonzero probability. Thus, the goal of the present game cannot be achieved. These complete the proof.
\end{proof}

Notice that to win the game, it is quite justified to start with the states of (\ref{eq2}) because if we choose any two states from these three states, then the two states are going to be linearly independent for sure when we take the coefficients $|a|, |b|, |a_1|, |a_2|$ as nonzero. In particular, we have taken here the states $\ket{\phi_1}$ and $\ket{\phi_2}$ as linearly independent. So, like $\{\ket{\phi_1}, \ket{\phi_2}\}$, $\{\ket{\phi_1^\perp}, \ket{\phi_2^\perp}\}$ are also linearly independent and they form a basis for two dimensional Hilbert space. Thus, $\ket{\phi_3}$ can be written as a linear combination of $\ket{\phi_1^\perp}$ and $\ket{\phi_2^\perp}$. This is what we have considered. However, we end up with the fact that if we want to achieve the maximum value of the average probability of success, then the game cannot be won. Again, there is no quantum advantage in maximizing the average probability of success as this maximum value can be achieved when Alice is sending a classical bit to Bob. In fact, from the proof of Proposition \ref{prop3}, it is clear that the maximum value of $\mathcal{P}_{avg}^{(3)}$ cannot be achieved when states of a quantum encoding strategy are pairwise linearly independent. Nevertheless, we are interested in the following: Bob unambiguously identifies the value of $X$ all the time (for all `$j$' values, where $j$ is associated with $S_j$), i.e., in every event with some nonzero probability. Furthermore, we search for the maximum value of $\mathcal{P}_{avg}^{(3)}$ when Alice and Bob win the game. From the preceding proposition, we can conclude that for winning the game $|b|$ must be nonzero. In that case we can start with the lower bound, given in (\ref{eq5}). In fact, we argue that this lower bound is achievable {\it if and only if} $(|a_1| + |a_2|) = 1$. We now put this if and only if condition in a proposition form.

\begin{proposition}\label{prop4}
The lower bound of (\ref{eq5}) is achievable if and only if $|a_1| + |a_2| = 1$, provided $|b|\neq0$.
\end{proposition} 

\begin{proof}
The `if' part is already shown in the proof of the preceding proposition, in particular, see (\ref{eq3}). Thus, for the `only if' part, we consider the following: 
\begin{equation}\label{eq7}
\begin{array}{c}
|a| + |b|(|a_1| + |a_2|) = |a| + |b|\\[1 ex]
\Rightarrow |b|(|a_1| + |a_2|-1) = 0.
\end{array}
\end{equation}
We have already mentioned that $|b|\neq0$, so, the only possibility is $(|a_1| + |a_2|-1) = 0$, i.e., $|a_1| + |a_2| = 1$, to satisfy the bound. These complete the proof. 
\end{proof}

We now want two things together: (a) winning the game and (b) achieving the lower bound of (\ref{eq5}). We set (a) here because it is the main goal. On the other hand, (b) helps to reduce the probability of inconclusive outcome. These together can be expressed through the following manner:
\begin{equation}\label{eq8}
\begin{array}{l}
\mbox{We want:}~~\epsilon >0,\\[1 ex]
\mbox{Such that:}~ (1-|\langle\phi_i|\phi_{i^\prime}\rangle|)\geq\epsilon,\\[1 ex] \forall i, i^\prime = 1,2,3,~ i\neq i^\prime\\[1 ex]
\mbox{and}~~~~~~~~ \\[1 ex]
|\langle\phi_1|\phi_2\rangle|+|\langle\phi_2|\phi_3\rangle|+|\langle\phi_3|\phi_1\rangle| = |a|+|b|.
\end{array}
\end{equation}
To solve the above, we can start with $|a_1| + |a_2| = 1$. This is to make $|\langle\phi_1|\phi_2\rangle|+|\langle\phi_2|\phi_3\rangle|+|\langle\phi_3|\phi_1\rangle| = |a|+|b|$. Now, $|a_1| + |a_2| = 1$ only when $a_1 = 0$ or $a_2 = 0$ because $|a|$ cannot be 1, see (\ref{eq4}) for details. Either of the conditions $a_1 = 0$ or $a_2 = 0$, provides us similar type of solution so, without loss of generality we can take $a_2 = 0$. Therefore, the states of (\ref{eq2}) become: 
\begin{equation}\label{eq9}
\begin{array}{c}
\ket{\phi_1} = \ket{\phi},~
\ket{\phi_2} = a\ket{\phi}+b\ket{\phi^\perp},~
\ket{\phi_3} = \ket{\phi^\perp}.
\end{array}
\end{equation}
We take $|a|\geq |b|$. So, for each values of `$j$' ($j$ is associated with $S_j$), the individual probabilities of success are 1, $1-|a|$, and $1-|b|$. Among these three probabilities, the minimum value is $1-|a|$. Thus, we take $1-|a|= \epsilon$. In this way, we get a solution of (\ref{eq8}), given by-
\begin{equation*}
\begin{array}{c}
\ket{\phi_1} = \ket{\phi},~~
\ket{\phi_2} = (1-\epsilon)\ket{\phi}+\sqrt{2\epsilon-\epsilon^2}\ket{\phi^\perp},\\[1 ex]
\ket{\phi_3} = \ket{\phi^\perp}.
\end{array}
\end{equation*}
In this case,
\begin{equation*}
\begin{array}{c}
|\langle\phi_1|\phi_2\rangle|+|\langle\phi_2|\phi_3\rangle|+|\langle\phi_3|\phi_1\rangle| = 1-\epsilon+\sqrt{2\epsilon-\epsilon^2},\\[1 ex]
\mbox{and}\\[1 ex]
\mathcal{P}_{avg}^{(3)} = \frac{1}{3}(2 + \epsilon - \sqrt{2\epsilon-\epsilon^2}).
\end{array}
\end{equation*} 
Let us understand the meaning of this solution with one example. 

{\it Example 1.}~We assume that $\epsilon = 0.1$. So, if we consider the encoding through the states $\ket{\phi_1} = \ket{\phi}$, $\ket{\phi_2} = 0.9\ket{\phi}+\sqrt{0.19}\ket{\phi^\perp}$, $\ket{\phi_3} = \ket{\phi^\perp}$, then $(1-|\langle\phi_i|\phi_{i^\prime}\rangle|)\geq0.1$ and $\mathcal{P}_{avg}^{(3)} = 0.5547$ (approx.). In fact, there is no solution for which $(1-|\langle\phi_i|\phi_{i^\prime}\rangle|)\geq0.1$ and at the same time, $\mathcal{P}_{avg}^{(3)} > 0.5547$. 

{\it Example 2.}~If we think about maximizing $\epsilon$ within the problem of (\ref{eq8}), then the only possibility is that we take $|a|=|b|=1/\sqrt{2}$. This is a special case of the problem of (\ref{eq8}). The maximum value of $\epsilon$ is given by- $1-(1/\sqrt{2})$, achievable through the states $\ket{\phi_1} = \ket{\phi}$, $\ket{\phi_2} = (\ket{\phi}+\ket{\phi^\perp})/\sqrt{2}$, $\ket{\phi_3} = \ket{\phi^\perp}$, $|\langle\phi_1|\phi_2\rangle|+|\langle\phi_2|\phi_3\rangle|+|\langle\phi_3|\phi_1\rangle| = \sqrt{2}$, and $\mathcal{P}_{avg}^{(3)} = 1-(\sqrt{2}/3)$. Notice that $|\langle\phi_1|\phi_2\rangle|+|\langle\phi_2|\phi_3\rangle|+|\langle\phi_3|\phi_1\rangle| = \sqrt{2}$ is the {\it greatest lower bound} and $\ket{\phi_2}$ is now a maximally coherent state\footnote{For details regarding quantum coherence, one can have a look into \cite{Streltsov17} and the references therein. However, we mention that all states of the form $\mu_0\ket{\mu_0^\prime}+\mu_1\ket{\mu_1^\prime}$ are coherent states with respect to the basis $\{\ket{\mu_0^\prime}, \ket{\mu_1^\prime}\}$, $|\mu_0|,|\mu_1|>0$. Now, suppose, this basis is an orthonormal basis and $\mu_0 = \mu_1 = 1/\sqrt{2}$, then the superposed states, just mentioned, are maximally coherent states with respect to the basis $\{\ket{\mu_0^\prime}, \ket{\mu_1^\prime}\}$.} with respect to $\{\ket{\phi}, \ket{\phi^\perp}\}$ basis.

\subsection{Role of the trine ensemble}
Having optimized the average probability of success, we now consider optimizing the individual probabilities of success. In the case of the problem of (\ref{eq8}), this means that we drop the constraint $|\langle\phi_1|\phi_2\rangle|+|\langle\phi_2|\phi_3\rangle|+|\langle\phi_3|\phi_1\rangle| = |a|+|b|$. Then, some of the individual probabilities of success might be improved. Alice and Bob can fix the following encoding process.
\begin{equation}\label{eq10}
\begin{array}{l}
x_1 \rightarrow \ket{\phi_1} = \ket{0},\\[1 ex]
x_2 \rightarrow \ket{\phi_2} = \frac{1}{2}(\ket{0}+\sqrt{3}\ket{1}),\\[1 ex]
x_3 \rightarrow \ket{\phi_3} = \frac{1}{2}(\ket{0}-\sqrt{3}\ket{1}).
\end{array}
\end{equation}
Accordingly, the measurement, which is performed by Bob to decode the information, is given by the POVM elements $\Pi_i$ = $\frac{2}{3}\ketbra{\tilde{\phi_i}}{\tilde{\phi_i}}$, where $\ket{\tilde{\phi_i}}$ is orthogonal to $\ket{\phi_i}$, defined in the above equation $\forall i = 1,2,3$. In this case, unambiguous discrimination is accomplished by elimination of a state. Here individual probability of success is 0.5 and the average probability of success is also 0.5. Clearly, in this case some of the individual probabilities of success are improved but the average probability of success is decreased. Then, the question of interest is: What is the optimal strategy such that the success probability in any individual case cannot be less than a maximal value?

\begin{equation*}
\begin{array}{l}
\mbox{We want:}~~\epsilon_{max} >0,\\[1 ex]
\mbox{Such that:}~ (1-|\langle\phi_i|\phi_{i^\prime}\rangle|)\geq\epsilon_{max},\\[1 ex] \forall i, i^\prime = 1,2,3,~ i\neq i^\prime,
\end{array}
\end{equation*}
where $\epsilon_{max}$ is the maximum probability at least achievable in the individual cases. To figure it out, we consider a specific encoding $x_i\rightarrow\ket{\phi_i}$, $i=1,2,3$. 
\begin{equation*}
\ket{\phi_1} = \ket{0},~\ket{\phi_2} = a\ket{0}+b\ket{1},~\ket{\phi_3} = a\ket{0}-b\ket{1},
\end{equation*}
where $a,b$ are complex numbers such that $|a|^2+|b|^2=1$. Notice that if we pick any two states from these three states then the picked states are linearly independent and thus, they can be distinguished unambiguously. One may think, for both states $\ket{\phi_2}$ and $\ket{\phi_3}$, why we consider same coefficients $a,b$. This is due to the following reason. Suppose, we take different coefficients for $\ket{\phi_3}$. Then, it may possible to reduce the overlap between $\ket{\phi_1}$ and $\ket{\phi_3}$ in comparison with the overlap between $\ket{\phi_1}$ and $\ket{\phi_2}$. But in that case the overlap between $\ket{\phi_2}$ and $\ket{\phi_3}$ may increase. On the other hand, if we take same coefficients then keeping the overlaps between the pairs $\{\ket{\phi_1}, \ket{\phi_2}\}$ and $\{\ket{\phi_1}, \ket{\phi_3}\}$ same, one can reduce the overlap between $\ket{\phi_2}$ and $\ket{\phi_3}$. This is extremely important to solve the above problem (finding the value of $\epsilon_{max}$). However, we do not consider $\ket{\phi_2}$ and $\ket{\phi_3}$ orthogonal to each other, as this case is already been discussed (see Example 2 of the previous subsection in this regard). Let us now calculate the probabilities of success for each pair of states.
\begin{equation*}
\begin{array}{lll}
j=1 \implies & p_{\{\ket{\phi_1}, \ket{\phi_2}\}} \implies & 1-|\langle\phi_1|\phi_2\rangle| = 1-|a|,\\[1 ex]
j=2 \implies & p_{\{\ket{\phi_2}, \ket{\phi_3}\}} \implies & 1-|\langle\phi_2|\phi_3\rangle| = 1-|(|a|^2-|b|^2)|\\[1 ex]
j=3 \implies & p_{\{\ket{\phi_3}, \ket{\phi_1}\}} \implies & 1-|\langle\phi_3|\phi_1\rangle| = 1-|a|
\end{array}
\end{equation*}
In this case, we assume $|a|<|b|$. So, $0<|a|^2<1/2$ and $1/2<|b|^2<1$. If $|a|^2=|b|^2=1/2$, then it is similar as Example 2. We now consider a small positive number $0<\delta<1/2$, such that $|a|^2 = 1/2-\delta$ and $|b|^2 = 1/2 + \delta$. So, the success probabilities, corresponding to $j=1,2,3$, become $1-\sqrt{1/2-\delta}$, $1-2\delta$, and $1-\sqrt{1/2-\delta}$. We then assume $2\delta\geq\sqrt{1/2-\delta}$ $\implies$ $\delta\geq1/4$. Clearly, $\epsilon_{max}=0.5$ when $\delta=1/4$. Similarly, if we take $2\delta\leq\sqrt{1/2-\delta}$, then, we get $\delta\leq1/4$. In this case also $\epsilon_{max}=0.5$ when $\delta=1/4$. These suggest the maximum achievable value of $\epsilon$, i.e., $\epsilon_{max}$ is 0.5 and this is happening when $\delta = 1/4$, thereby, the values of $|a|^2$ and $|b|^2$ are 1/4 and 3/4. Thus, we get trine ensemble as the optimal solution when we only focus on improving the individual probabilities. One can also check by considering $|a|>|b|$ but no better value of $\epsilon$ can be obtained with this consideration.

From the above discussion, it is clear that finding the value of $\epsilon_{max}$ is connected with finding the point where all of the individual probabilities are equal. So, in other words, if we set the overlap of $\ket{\phi_1}$ and $\ket{\phi_2}$ as $|a|$, then clearly the overlap of $\ket{\phi_1}$ and $\ket{\phi_3}$ also needs to be $|a|$. This gives the construction of the encoding presented here, and rest is to find the point where $1 - |a| = 1 - |(|a|^2 - |b|^2)|$, which occurs when $|a| = 1/2$ and $|b| = \sqrt{3}/2$. This solution gives the states of trine ensemble (for details regarding the trine ensemble see Ref.~\cite{Barnett09} and the references therein).

Furthermore, notice that in case of a quantum strategy like Proposition \ref{prop2}, for which the game can be won, Bob can change his measurement according to the information of the set that he receives from Referee. For example, if Bob receives ``1'', then Bob distinguishes between $\ket{\phi_1}$ and $\ket{\phi_2}$. Again, if Bob is given ``2'', then Bob distinguishes between $\ket{\phi_2}$ and $\ket{\phi_3}$. In each case, Bob can choose a measurement defined by suitable positive operator valued measure (POVM) elements to achieve the optimal probability. Clearly, if the following constraint is put on Bob that `he cannot change his measurement', i.e., he is allowed to perform only one measurement, then the situation is more complex. Here the motivation of fixing this constraint can be described as the following. Actually, in our case, a measurement corresponds to a specific setup. So, with increasing the number of required measurements, the number of required setups also increases. This is certainly a costly affair. Therefore, it is reasonable to consider the constraint that `Bob is not able to change his measurement'. In fact, it demonstrates less resource requirement in a practical situation as Bob is allowed to use only one measurement setup. However, for the simplest case (i.e., $d=3$) which is described above, we can have a solution through the so-called `trine' ensemble.

\section{Four dimensional random variable}\label{sec3}
In this section, we consider that the dimension of the random variable is four, i.e., $d=4$ but still we consider that $n=2$. So, in this case there could be six random sets, given by - $S_1 = \{x_1 , x_2\}$, $S_2 = \{x_1 , x_3\}$, $S_3 = \{x_1 , x_4\}$, $S_4 = \{x_2 , x_3\}$, $S_5 = \{x_2 , x_4\}$, $S_6 = \{x_3 , x_4\}$. Ultimately, Bob receives a `$j$' value (`$j$' can be $1,2,\dots,6$) from the Referee and try to distinguish between two states, where the encoding of the quantum scenario is like $x_i\rightarrow\ket{\phi_i},~\forall i = 1,2,3,4$. Overall, here our communication task goes on in the same way as it is described in the previous section.

In this case also if Alice sends only a cbit then the game cannot be won. This is obvious because when $d=3$ by sending a cbit, the game cannot be won, then when $d$ is increased but $n$ is the same, the complexity of the game is also increased. Therefore, it is obvious that the game cannot be won. However, we will show that by sending a qubit, the game can be won. We define $\mathcal{P}_{avg}^{(4)}$ as the average probability of success when $d=4$. $\mathcal{P}_{avg}^{(4)}$ is given by-
\begin{equation}\label{eq11}
\begin{array}{c}
\mathcal{P}_{avg}^{(4)} = 1-\frac{1}{6}[|\langle\phi_1|\phi_2\rangle|+|\langle\phi_1|\phi_3\rangle|+|\langle\phi_1|\phi_4\rangle|\\[1 ex]+|\langle\phi_2|\phi_3\rangle|+|\langle\phi_2|\phi_4\rangle|+|\langle\phi_3|\phi_4\rangle|]
\end{array}
\end{equation}
We rewrite this equation in the following way:
\begin{equation}\label{eq12}
\begin{array}{c}
\mathcal{P}_{avg}^{(4)} = 1-\frac{1}{12}[(|\langle\phi_1|\phi_2\rangle|+|\langle\phi_2|\phi_3\rangle|+|\langle\phi_1|\phi_3\rangle|)\\[1 ex]+(|\langle\phi_2|\phi_3\rangle|+|\langle\phi_2|\phi_4\rangle|+|\langle\phi_3|\phi_4\rangle|)\\[1 ex]+(|\langle\phi_1|\phi_2\rangle|+|\langle\phi_1|\phi_4\rangle|+|\langle\phi_2|\phi_4\rangle|)\\[1 ex]+(|\langle\phi_1|\phi_3\rangle|+|\langle\phi_1|\phi_4\rangle|+|\langle\phi_3|\phi_4\rangle|)]\\[2 ex]
\Rightarrow \mathcal{P}_{avg}^{(4)} = \frac{1}{4}[(\mathcal{P}_{avg}^{(3)})_{123}+(\mathcal{P}_{avg}^{(3)})_{234}+(\mathcal{P}_{avg}^{(3)})_{124}\\[1 ex]+(\mathcal{P}_{avg}^{(3)})_{134}],
\end{array}
\end{equation}
where $(\mathcal{P}_{avg}^{(3)})_{klm}$ stands for the average probability of success, provided there are three values of the random variable available $x_k$, $x_l$, and $x_m$, $k\neq l \neq m$. Clearly, $\mathcal{P}_{avg}^{(4)}$ is maximum when individual $(\mathcal{P}_{avg}^{(3)})_{klm}$ are maximum. From the previous section, it is known that $(\mathcal{P}_{avg}^{(3)})_{klm}\leq \frac{2}{3}$ and the equality holds when we apply an encoding like: $x_k\rightarrow0$, $x_l\rightarrow0$, and $x_m\rightarrow1$. Therefore, it must be the case that $\mathcal{P}_{avg}^{(4)}\leq\frac{2}{3}$ but the question is if an encoding exists for which $\mathcal{P}_{avg}^{(4)}=\frac{2}{3}$. Such an encoding is given by $x_1\rightarrow0$, $x_2\rightarrow0$, $x_3\rightarrow1$, and $x_4\rightarrow1$. Notice that if we choose any three values from $\{x_i|i=1,2,3,4\}$, the encoding is always like $x_k\rightarrow0$, $x_l\rightarrow0$, and $x_m\rightarrow1$ or, like $x_k\rightarrow1$, $x_l\rightarrow1$, and $x_m\rightarrow0$, $k\neq l\neq m$. 

However, for this type of encoding if Bob receives `1' or `6', he fails to identify the value of the random variable unambiguously with some nonzero probability. This is not allowed if Alice and Bob want to win the game. At the same time it is also important to maximize $\mathcal{P}_{avg}^{(4)}$. As we have seen, to maximize $\mathcal{P}_{avg}^{(4)}$, we have to maximize the individual quantities $(\mathcal{P}_{avg}^{(3)})_{klm}$. Ultimately, it is required to minimize the quantity $|\langle\phi_k|\phi_l\rangle|+|\langle\phi_l|\phi_m\rangle|+|\langle\phi_m|\phi_k\rangle|$, for $k,l,m\in\{1,2,3,4\}$ and $k\neq l\neq m$ and thus, we can think about the lower bound of (\ref{eq5}). We are now ready to present a similar problem as given in (\ref{eq8}) but for $d=4$.
\begin{equation}\label{eq13}
\begin{array}{l}
\mbox{We want:}~~\epsilon >0,\\[1 ex]
\mbox{Such that:}~ (1-|\langle\phi_i|\phi_{i^\prime}\rangle|)\geq\epsilon,\\[1 ex] \forall i, i^\prime = 1,2,3,4,~ i\neq i^\prime\\[1 ex]
\mbox{and}~~~~~~~~ \\[1 ex]
|\langle\phi_k|\phi_l\rangle|+|\langle\phi_l|\phi_m\rangle|+|\langle\phi_m|\phi_k\rangle| = |a|+|b|,\\[1 ex]
\forall k,l,m = 1,2,3,4, k\neq l\neq m. 
\end{array}
\end{equation}
One may think that for different sets of $\{k,l,m\}$, there should be different sets of $\{a,b\}$. But it is not the case. In particular, we will fix $\epsilon$ first then we will fix the condition following the previous section, i.e., $|\langle\phi_k|\phi_l\rangle|+|\langle\phi_l|\phi_m\rangle|+|\langle\phi_m|\phi_k\rangle| = 1-\epsilon+\sqrt{2\epsilon - \epsilon^2}$. Remember that we want to achieve the lower bound of (\ref{eq5}) as it helps to maximize the individual $(\mathcal{P}_{avg}^{(3)})_{klm}$ and thereby $\mathcal{P}_{avg}^{(4)}$. Now, for achieving lower bound of (\ref{eq5}), the if and only if condition is $|a_1| + |a_2| = 1$. Ultimately, the solution of the problem (\ref{eq13}) turns out to be the following:
\begin{equation}\label{eq14}
\begin{array}{c}
\ket{\phi_1} = \ket{\phi},~
\ket{\phi_2} = (1-\epsilon)\ket{\phi}+\sqrt{2\epsilon-\epsilon^2}\ket{\phi^\perp},\\[1 ex]
\ket{\phi_3} = \ket{\phi^\perp},
\ket{\phi_4} = \sqrt{2\epsilon-\epsilon^2}\ket{\phi}-(1-\epsilon)\ket{\phi^\perp}.
\end{array}
\end{equation}
Now, if we choose any three states from the above then the condition $|\langle\phi_k|\phi_l\rangle|+|\langle\phi_l|\phi_m\rangle|+|\langle\phi_m|\phi_k\rangle| = 1-\epsilon+\sqrt{2\epsilon - \epsilon^2}$ is satisfied. It is also easily verifiable that $(1-|\langle\phi_i|\phi_{i^\prime}\rangle|)\geq\epsilon$, $\forall~i, i^\prime = 1,2,3,4,~ i\neq i^\prime$. In this regard, note that we have assumed $|a|\geq|b|$, so, $(1-\epsilon)\geq\sqrt{2\epsilon-\epsilon^2}$, resulting the fact that $1-|\langle\phi_2|\phi_3\rangle|$ = $1-|\langle\phi_1|\phi_4\rangle|$ = $1-\sqrt{2\epsilon - \epsilon^2}\geq\epsilon$. Notice that $\forall i = 1,2,3,4$, $\ket{\phi_i}$s of (\ref{eq14}) belong to two dimensional Hilbert space and thus, by sending a qubit from Alice's side to Bob the problem of (\ref{eq13}) can be solved and thereby, the game can be won for $d=4$ and $n=2$.

\subsection*{Role of mutually unbiased bases}
{\it Example 3.}~If we think about maximizing $\epsilon$ within the problem of (\ref{eq13}), then the only possibility is that we take $|a|=|b|=1/\sqrt{2}$, i.e., $(1-\epsilon) = \sqrt{2\epsilon-\epsilon^2} = 1/\sqrt{2}$. This is a special case of the problem of (\ref{eq13}). The maximum value of $\epsilon$ is given by $1-(1/\sqrt{2})$, achievable through the encoding $x_i\rightarrow\ket{\phi_i},~\forall i = 1,2,3,4$, where the states $\ket{\phi_1} = \ket{\phi},~
\ket{\phi_2} = (\ket{\phi} + \ket{\phi^\perp})/\sqrt{2},~
\ket{\phi_3} = \ket{\phi^\perp}$, $\ket{\phi_4} = (\ket{\phi}-\ket{\phi^\perp})/\sqrt{2}$, $|\langle\phi_k|\phi_l\rangle|+|\langle\phi_l|\phi_m\rangle|+|\langle\phi_m|\phi_k\rangle| = \sqrt{2}$, and $\mathcal{P}_{avg}^{(3)} = \mathcal{P}_{avg}^{(4)} = 1-(\sqrt{2}/3)$. Notice that $\ket{\phi_i}$s are from mutually unbiased bases of two dimensional Hilbert space. We now put this observation in a proposition form.
\begin{proposition}\label{prop5}
To maximize $\epsilon$ in the problem of (\ref{eq13}), it is necessary and also sufficient to encode the values of a four dimensional random variable, within the states of mutually unbiased bases\footnote{We consider two bases $\mathcal{B}_1$ = \{$\ket{\nu_1}$, $\ket{\nu_2}$, \dots, $\ket{\nu_{\mathcal{D}}}$\} and $\mathcal{B}_2$ = \{$\ket{\nu_1^\prime}$, $\ket{\nu_2^\prime}$, \dots, $\ket{\nu_{\mathcal{D}}^\prime}$\} in a ${\mathcal{D}}$-dimensional Hilbert space. We say that these bases are mutually unbiased if and only if $|\langle\nu_i|\nu_{i^\prime}^\prime\rangle|$ = $1/\sqrt{{\mathcal{D}}}$, for every $i,i^\prime$. For details regarding mutually unbiased bases, one can have a look into \cite{Ivonovic81, Ivanovic97, Bandyopadhyay02}.}.
\end{proposition}

The proof of the above proposition follows from {\it Example 2} and {\it Example 3}.

\section{Higher dimensional random variable}\label{sec4}
In this section, we consider that the dimension of the random variable is $d\geq5$. But still we consider that $n=2$. Ultimately, Bob receives a `$j$' value from the Referee and then try to distinguish between two states, where the encoding in the quantum scenario is given by $x_i\rightarrow\ket{\phi_i},~\forall i \in \{1,2,\dots,d\}$. Following similar argument as given in the previous section, it can be argued that by sending only a cbit, the game cannot be won when $d\geq5$ too. We can define $\mathcal{P}_{avg}^{(d)}$ in the similar way as we did for $\mathcal{P}_{avg}^{(4)}$ and $\mathcal{P}_{avg}^{(3)}$. $\mathcal{P}_{avg}^{(d)}$ is given by-
\begin{equation}\label{eq15}
\begin{array}{c}
\mathcal{P}_{avg}^{(d)} = 1 - \frac{2}{d(d-1)}(\sum_{i, i^\prime}|\langle\phi_i|\phi_{i^\prime}\rangle|)\\[2 ex]

= \frac{6}{d(d-1)(d-2)}[\sum_{k,l,m}(\mathcal{P}_{avg}^{(3)})_{klm}],
\end{array}
\end{equation}
where $i,i^\prime,k,l,m\in\{1,2,\dots,d\}$, $i\neq i^\prime$ and $k\neq l\neq m$. $(\mathcal{P}_{avg}^{(3)})_{klm}$ are similar quantities as defined in the previous section. The second line of the above equation tells us that if the quantities $(\mathcal{P}_{avg}^{(3)})_{klm}$ are maximized then the quantity $\mathcal{P}_{avg}^{(d)}$ will be maximized. We know that $(\mathcal{P}_{avg}^{(3)})_{klm}$ can be maximized when we use the following (classical) encoding $x_k\rightarrow0$, $x_l\rightarrow0$, and $x_m\rightarrow1$ (or, in other words, any two of the three randomly chosen values of a random variable are encoded against the same bit value and the third value of the random variable is encoded against the orthogonal bit value). It is easy to check that when $d\geq5$, it is not possible to have an encoding strategy such that among $d$ values of a random variable if we randomly choose three values $x_k$, $x_l$, and $x_m$ then corresponding quantity, $(\mathcal{P}_{avg}^{(3)})_{klm}$ is maximum. Nevertheless, what could be a sensible choice here is that we can fix a number $N<\frac{d!}{3!(d-3)!}$ for $d\geq5$ and we can maximize this number $N$. The significance of this number is that we can get at least $N$ ensembles of randomly chosen $\{x_k, x_l, x_m\}$ such that the quantity $(\mathcal{P}_{avg}^{(3)})_{klm}$ is maximum. So, the question is: What is the encoding strategy corresponding to maximum $N$? We first adopt an encoding strategy for even $d$, where half of the values of the random variable are encoded against 0 and the other half against 1. Here, $N$ = $\frac{d!}{3!(d-3)!} - 2\cdot\frac{(d/2)!}{3!(d/2-3)!}$. If we encode $(d/2+d^\prime)$ values of the random variable against a particular bit value and the remaining values of the random variable against the orthogonal bit value, then, the value of $N$ becomes strictly less than $\frac{d!}{3!(d-3)!} - 2\cdot\frac{(d/2)!}{3!(d/2-3)!}$ for any nonzero $d^\prime$ ($d^\prime$ is an integer). In this way, we argue that the strategy, where half of the values of the random variable are encoded against a particular bit value and the other half against the orthogonal bit value, is the best strategy, i.e., $\mathcal{P}_{avg}^{(d)}$ is maximized here, when $d$ is even. In the same way (as argued for even $d$), it is possible to show that for odd $d$ one can adopt a strategy where $(d-1)/2+1$ values of the random variable can be encoded against a particular bit value while $(d-1)/2$ values of the random variable can be encoded against the orthogonal bit value. In fact, this strategy is best strategy for odd $d$ in a sense that for this strategy $N$ is going to be maximum and thereby, the quantity $\mathcal{P}_{avg}^{(d)}$ will be maximum. However, we have to remember that these strategies do not help to win the game. 

We can also define a similar problem as given in (\ref{eq13}) for $d\geq5$. But not for all randomly chosen $x_k, x_l, x_m$ the relation $|\langle\phi_k|\phi_l\rangle|+|\langle\phi_l|\phi_m\rangle|+|\langle\phi_m|\phi_k\rangle| = |a|+|b|$ can be true. This can be easily checked. So, again we have to consider the number $N$ and we have to maximize $N$, such that at least we can get $N$ ensembles of randomly chosen $x_k,x_l,x_m$ for which the relation $|\langle\phi_k|\phi_l\rangle|+|\langle\phi_l|\phi_m\rangle|+|\langle\phi_m|\phi_k\rangle| = |a|+|b|$ can be true. This modified version of the problem is particularly important if we try to maximize $\epsilon$ when $d=5$ or, $d=6$. Because, in these cases, one can consider encoding the values $x_5$ and $x_6$ against the eigenvectors of the Pauli matrix $\sigma_y$ to solve the problem. In higher dimensions $(d>6)$, it is not known how to solve this problem with maximum $\epsilon$. However, if we just think about winning the game dropping the condition of satisfying the relation $|\langle\phi_k|\phi_l\rangle|+|\langle\phi_l|\phi_m\rangle|+|\langle\phi_m|\phi_k\rangle| = |a|+|b|$, then it is possible by sending only a qubit even if $d>6$. The reason is described in the following.

\subsection*{Large quantum-classical separation} 
Suppose, the dimension of the random variable, the value of which is sent to Alice, is `$d>6$'. [Previously, we have discussed how to maximize $\mathcal{P}_{avg}^{(d)}$ for $d\geq5$. But with that strategy the game cannot be won. Here we want to discuss a strategy for $d>6$, with which the game can be won.] But ultimately, Referee sends the information of a random set of cardinality two to Bob. Here also, Alice is allowed to send Bob one (qu)bit of information. In this scenario, we can establish quite high advantage of quantum communication over its classical counterpart. In brief, we say this as `large quantum-classical separation'. However, we mention that this advantage is demonstrated with respect to a specific goal, i.e., with some non-zero probability, Bob has to identify the value of the random variable, sent to Alice by the referee no matter which value she receives or what information Bob receives from the referee. This `high advantage' is explained in a later portion. This is based on the fact that two quantum states $\ket{0}$ and $\ket{1}$ can be superposed in infinitely many ways. We suppose that the superposed states are $a_i\ket{0}+b_i\ket{1}$, where $|a_i|^2+|b_i|^2=1$ and $|a_i|$, $|b_i|$ are nonzero. It is also possible to ensure that for any value of $d$, encoding process can be done in a way if any two states are chosen then they must be linearly independent. The linear independence part is to confirm the unambiguous identification of the value of the random variable. 

Here the parties apply the following encoding process: $x_i \rightarrow a_i\ket{0}+b_i\ket{1}$, $\forall i=1,2,\dots,d$ and both $a_i,b_i$ are nonzero. For simplicity, we can take them as positive. Next, we assume that between two arbitrary states $a_l\ket{0}+b_l\ket{1}$ and $a_{l^\prime}\ket{0}+b_{l^\prime}\ket{1}$, Bob has to distinguish unambiguously with some nonzero probability. So, we have to find out the condition for which these two states can be linearly independent. We take $c_1(a_l\ket{0}+b_l\ket{1})+c_2(a_{l^\prime}\ket{0}+b_{l^\prime}\ket{1})\equiv(0,0)$ or, $(c_1a_l+c_2a_{l^\prime})\ket{0}+(c_1b_l+c_2b_{l^\prime})\ket{1}\equiv(0,0)$. Thus, $(c_1a_l+c_2a_{l^\prime}) = 0 = (c_1b_l+c_2b_{l^\prime})$. Clearly, for different positive values of $a_l, a_{l^\prime}, b_l, b_{l^\prime}$, both $(c_1a_l+c_2a_{l^\prime})$ and $(c_1b_l+c_2b_{l^\prime})$ are zero when $c_1=c_2=0$ or $\frac{a_l}{b_l} = \frac{a_{l^\prime}}{b_{l^\prime}}$. But the second condition does not arise if $a_l$ and $a_{l^\prime}$ are different since $a_i^2+b_i^2=1$ $\forall i = l, l^\prime$. Therefore, the only option which is left is $c_1=c_2=0$. This implies that the states $a_l\ket{0}+b_l\ket{1}$ and $a_{l^\prime}\ket{0}+b_{l^\prime}\ket{1}$ are linearly independent and they can be distinguished unambiguously with some nonzero probability. In this way, we can construct a strategy of sending a qubit, via which it is always possible to identify the value of the random variable unambiguously with some nonzero probability under the present conditions.

On the other hand, by sending a cbit it is not possible to identify the value of the random variable unambiguously with some nonzero probability under the present conditions. The proof is due to the similar argument as given in the proof of Proposition \ref{prop1}. Moreover, recall that `$d$' can be anything. Arguably, for a very large `$d$', to accomplish the present task Alice must send a large number of classical bits. Therefore, in the present scenario, a qubit is always effective but a large number of cbits may not be. In this way, one can realize a large quantum-classical separation. In fact, theoretically this separation can be arbitrary large. However, if Alice sends only a qubit for winning the game the success probability of unambiguously identifying the value of the random variable with increasing $d$, must be decreasing. Clearly, for a large value of $d$, it may not be possible to quantify the small value of success probability (individual or average) experimentally. Again, how far, these success probabilities can be determined experimentally, is a completely different problem and we are leaving it for future studies. 

\section{General description of the game}\label{sec5}
The general description of the game is given as the following. There are three spatially separated parties, Referee, Alice, and Bob. The Referee sends a value $x_i\in\{x_1,x_2,\dots,x_d\}$ of a random variable $X$ to Alice, where $d$ is the dimension of $X$. At the same time, Referee also sends the information $j$ of a random set $S_j$ to Bob such that $S_j$ contains the particular $x_i$ which is sent to Alice, along with some other value(s) $x_{i^\prime}\in\{x_1,x_2,\dots,x_d\}$, but $i\neq i^\prime$. Note that Alice does not know the information $j$ which is sent to Bob by the Referee and similarly, Bob does not know the information of $x_i$ which is sent to Alice by the Referee. The task of Bob is to identify the value of the random variable which is sent to Alice by the Referee. Clearly, the question of interest is if the task can be accomplished for any value of $j$. However, to help Bob, Alice sends $n$-level information to Bob regarding $x_i$. This communication is one-way, i.e., there is no communication from Bob's side to Alice's side. But before the game starts, they (Alice and Bob) can fix an encoding-decoding strategy. Notice that if $n=d$ then the scenario is trivial, i.e., Bob is able to identify the value of the random variable perfectly (with 100 $\%$ certainty) for any value of $j$. When $n<d$ (i.e., $n$ is limited), Bob is not able to identify the value of the random variable perfectly for all values of $j$. Then, it can be explored, how well Bob can identify the value of the random variable for any value of $j$. 

We mention that for unambiguous identification, we have to keep the size of the set $S_j$, i.e., $|S_j|=n$, $2\leq n<d$. Because if $|S_j|>n$, then, unambiguous identification of the value of the random variable is clearly not possible. Now, when a value of the random variable is sent to Alice, information of a set is sent to Bob. This set must contain the value which is sent to Alice along with other $n-1$ values. Here the question is how many such sets are possible? This is clearly equal to $\binom{d}{n} = \frac{d!}{n!(d-n)!}$. These numbers define different values of $j$. Remember that Alice does not know the information of which set is sent to Bob, but she knows that the set must contain the value the random variable, which she has received, along with some other value(s). We mention that the set of values of the variable $X$, i.e., $\{x_1, x_2, ..., x_d\}$ is known to both Alice and Bob. The value of $n$ is also known to Alice and Bob. Based on the values of $d$ and $n$, several random sets can be defined, these definitions are also known to them. 

\subsection*{n $>$ 2 case}
We next consider that when the dimension of the random variable is `$d$', Referee sends a random set of cardinality `$n$' to Bob. In this case, if $n>2$, then sending a qubit from Alice's side to Bob will not help in accomplishing the task. In particular, it is possible to show that the above is solvable if Alice is allowed to send a qunit ($n$-level quantum system) to Bob. On the other hand, sending a cnit ($n$-level classical system) will not help in accomplishing the task of identifying the value of the random variable unambiguously with some nonzero probability for all $j$ values. We mention that in the classical case, an $n$-level information is defined by cnit which can have the values 0, 1, ..., $n-1$. In the quantum case, an $n$-level information is defined by qunit which can have the states $\ket{0}$, $\ket{1}$, ..., $\ket{n-1}$. The general case is quite straightforward. Here we only discuss the case when $d=4$ and $n=3$, that is, Alice is given $x_i\in\{x_1, x_2, x_3, x_4\}$ and Bob is given `$j$' where `$j$' is associated with $S_j$, $\forall j=1,2,3,4$. Here, $S_1 = \{x_1, x_2, x_3\}$, $S_2 = \{x_1, x_3, x_4\}$, $S_3 = \{x_1, x_2, x_4\}$, and $S_4 = \{x_2, x_3, x_4\}$.

We assume that Alice is allowed to send a ctrit (three-level classical system) to Bob. The values of the random variable are equally probable and thus, the sets $S_j$ are also equally probable. So, here Alice has to compute a function which must output different values for different $x_i$, otherwise, unambiguous identification of $x_i$ is impossible for all values of $j$. Now, even if computation of such a function is possible, encoding the values of the function corresponding to different $x_i$, within a three-level classical system is impossible. Thus, it is not possible for Alice and Bob to win that game when Alice is allowed to send only a ctrit to Bob.

However, it is possible to construct a quantum strategy via which Bob can identify the value of the random variable unambiguously with some nonzero probability for all values of $j$ when Alice is sending only a qutrit to Bob. The encoding strategy is given as the following: $x_i \rightarrow \ket{\phi_i}$, $\forall i=1,2,3,4$. We can take $\ket{\phi_i}$ as linearly independent states $\forall i=1,2,3$ and we can take $\ket{\phi_4}$ as $a_1\ket{\phi_1}+a_2\ket{\phi_2}+a_3\ket{\phi_3}$, where $|a_i|$ are nonzero. $a_i$ are chosen in such a way that $\ket{\phi_4}$ must be a valid state. For simplicity, one can simply take $\ket{\phi_i}$ as $\ket{i}$, $i=1,2,3$. Here, $\{\ket{i}\}$ forms a basis for qutrit system. Now, notice that for any value of `$j$', Bob is left with three linearly independent vectors which can be distinguished unambiguously with some nonzero probability. Therefore, the value of the random variable can be identified unambiguously with some nonzero probability for all values of $j$. We mention that if $\{\ket{\phi_1}, \ket{\phi_2}, \ket{\phi_3}\}$ is a basis, then we say that the states of the form $\mu_1\ket{\phi_1}+\mu_2\ket{\phi_2}+\mu_3\ket{\phi_3}$ is a coherent state of coherence rank three with respect to the considered basis, here $|\mu_i|>0$ $\forall i$. We now provide the following proposition:

\begin{proposition}\label{prop6}
For winning the game, it is necessary and also sufficient that the state $\ket{\phi_4}$ must have coherence rank three with respect to the basis $\{\ket{\phi_1}, \ket{\phi_2}, \ket{\phi_3}\}$.
\end{proposition}

\begin{proof}
The states $\ket{\phi_1}$, $\ket{\phi_2}$, and $\ket{\phi_3}$ are linearly independent. So, they form a basis for a qutrit (three-level quantum system) system. The sufficient condition follows from the fact that there exists a strategy which is given above. The necessary condition follows from a couple of arguments. If coherence rank is less than three then there is at least one value of `$j$' for which three states (to be distinguished by Bob) are not linearly independent and thus, the unambiguous discrimination of such states is not possible. Furthermore, coherence rank cannot be greater than three when Alice is sending a qutrit to Bob. These complete the proof.
\end{proof}

From the above example, it is quite realizable that using the same model it is possible to show -- a qunit is more powerful resource than a cnit in the context of the present task. Again, if the dimension of the random variable `$d$' is very large then also it is possible to show that a qunit can provide advantage over a large number of cnits in the context of achieving a specific goal in our game. 

\section{Conclusion}\label{sec6}
To develop quantum technologies, it is necessary to explore what advantages one can achieve using quantum resources over their classical counterparts. Furthermore, it is also important to identify the scenarios where it is not possible to achieve such advantages. 

In this work, we have designed a task which can be described in terms of a game, mediated by Referee and played between Alice and Bob. Referee sends Alice a value of a random variable. At the same time, Referee also sends Bob some partial information regarding that value. Here partial information can be defined in the following way. Bob gets the information of a random set which must contain the value of the variable, that is sent to Alice by the Referee, along with other value(s). Alice is not allowed to know what information is sent to Bob by the Referee. Again, Bob does not know which value of the random variable is sent to Alice. Now, the game can be won if and only if Bob can unambiguously identify the value of the variable with some nonzero probability, no matter what information Bob receives or which value is sent to Alice. However, to help Bob, Alice sends some limited information to him based on any pre-decided strategy.

For this game, we have shown an advantage of sending a qubit over cbit(s). However, whether there is at all any quantum advantage, depends on the goal, we set. In particular, we have proved that in some scenarios, it is never possible to achieve any quantum advantage. We also mention that to establish quantum advantage, it is not necessary to share entanglement among the spatially separated parties in the present game. Actually, here quantum coherence is playing the key role. We further have analyzed several variants of the game and provided certain bounds on the success probabilities. Moreover, we have established connections between trine ensemble, mutually unbiased bases, and the encoding-decoding strategies of the variants. In fact, our games should be treated as applications of trine ensemble, mutually unbiased bases, and quantum coherence.

To understand the application of the present game, it is required to explain its similarity with quantum dense coding protocol \cite{Bennett92}. Suppose, we consider the simplest case of our game, i.e., $d =3$ case. In this case, Alice is given a random two-bit string which belongs to the set $\{00, 01, 10\}$. One can think that these are basically values of the random variable. There is a limited communication from Alice's side to Bob which is one qubit or one cbit. The task of Bob is to identify the bit string error-freely. The only difference between the dense coding and our game, is that in the former protocol there is entanglement present between Alice and Bob while in our case, there is no entanglement present between Alice and Bob. Instead, Bob is receiving an additional information (which is from the Referee) in our case. Like dense coding, here also, when Alice communicates a qubit, more information can be extracted by Bob regarding the bit string of Alice. Now, the setting of dense coding is well established in quantum information theory. Therefore, exploiting its connection with our game, one can think about various applications in information processing protocols. However, further analysis is required to exhibit such applications explicitly.

Finally, we want to talk about the experimental realization of the game. As explained in the above, our game has similarity with the setting of dense coding protocol. In fact, the dense coding protocol was experimentally demonstrated several years back \cite{Mattle96}. So, we believe that there is a possibility to demonstrate our games experimentally. However, we mention that here we have considered a probabilistic setting. In this regard, we mention about Ref.~\cite{Webb23} where optimal unambiguous state elimination problem has been demonstrated experimentally. Now, in our case when $n=2$, state elimination and state discrimination are equivalent. So, there is a possibility of demonstrating some versions of our game experimentally. Nevertheless, with increasing dimension of the random variable the situation will become more complex.

For further research, we leave the following open questions. What will happen in our communication tasks when extra resources like randomness, entanglement etc. are provided between Alice and Bob?

\section*{Acknowledgments}
We acknowledge discussion with Michał Parniak, Michał Lipka, and Mateusz Mazelanik. This work was supported by the National Science Centre, Poland (Grant No. 2022/46/E/ST2/00115) and the ``Quantum Optical Technologies'' project, carried out within the International Research Agendas programme of the Foundation for Polish Science co-financed by the European Union under the European Regional Development Fund. M.B. acknowledges funding from the National Mission in Interdisciplinary Cyber-Physical systems from the Department of Science and Technology through the I-HUB Quantum Technology Foundation (Grant no: I-HUB/PDF/2021-22/008), support through the research grant of INSPIRE Faculty fellowship from the Department of Science and Technology, Government of India, and the start-up research grant from SERB, Department of Science and Technology (Grant no: SRG/2021/000267).

\bibliographystyle{apsrev4-2}
\bibliography{ref}
\end{document}